\documentclass[11pt,onecolumn]{IEEEtran}

\IEEEoverridecommandlockouts                              

\usepackage{graphicx} 
\usepackage{amsmath} 
\usepackage{amssymb}  

\usepackage{mathtools}
\usepackage{multicol}
\usepackage{amsthm}

\newtheorem{theorem}{Theorem}[section]
\newtheorem{lemma}{Lemma}[section]

\newtheorem{corollary}{Corollary}[section]

\newenvironment{definition}[1][Definition]{\begin{trivlist}
\item[\hskip \labelsep {\bfseries #1}]}{\end{trivlist}}

\newenvironment{remark}[1][Remark]{\begin{trivlist}
\item[\hskip \labelsep {\bfseries #1}]}{\end{trivlist}}

\newcommand{\qedd}{\nobreak \ifvmode \relax \else
      \ifdim\lastskip<1.5em \hskip-\lastskip
      \hskip1.5em plus0em minus0.5em \fi \nobreak
      \vrule height0.75em width0.5em depth0.25em\fi}

\makeatletter
\renewcommand*\env@matrix[1][*\c@MaxMatrixCols c]{%
  \hskip -\arraycolsep
  \let\@ifnextchar\new@ifnextchar
  \array{#1}}
\makeatother

\title{\LARGE \bf
New Guarantees for Blind Compressed Sensing
}

\author{Mohammad Aghagolzadeh and Hayder Radha
\thanks{Mohammad Aghagolzadeh and Hayder Radha are with Department of Electrical and Computer Engineering, Michigan State University, East Lansing, MI, USA. Email: aghagol1@msu.edu and radha@msu.edu.}%
}

\begin{document}

\maketitle
\thispagestyle{empty}
\pagestyle{empty}

\begin{abstract}

Blind Compressed Sensing (BCS) is an extension of Compressed Sensing (CS) where the optimal sparsifying dictionary is assumed to be unknown and subject to estimation (in addition to the CS sparse coefficients). Since the emergence of BCS, dictionary learning, a.k.a. sparse coding, has been studied as a matrix factorization problem where its sample complexity, uniqueness and identifiability have been addressed thoroughly. However, in spite of the strong connections between BCS and sparse coding, recent results from the sparse coding problem area have not been exploited within the context of BCS. In particular, prior BCS efforts have focused on learning constrained and complete dictionaries that limit the scope and utility of these efforts. In this paper, we develop new theoretical bounds for perfect recovery for the general \textit{unconstrained} BCS problem. These unconstrained BCS bounds cover the case of overcomplete dictionaries, and hence, they go well beyond the existing BCS theory. Our perfect recovery results integrate the combinatorial theories of sparse coding with some of the recent results from low-rank matrix recovery. In particular, we propose an efficient CS measurement scheme that results in practical recovery bounds for BCS. Moreover, we discuss the performance of BCS under polynomial-time sparse coding algorithms.

\end{abstract}

\section{Introduction}

The \textit{sparse representation} problem involves solving the system of linear equations $y=Ax\in\mathbb{R}^{d}$ where $x\in\mathbb{R}^m$ is assumed to be $k$-sparse; i.e. $x$ is allowed to have (at most) $k$ non-zero entries. The matrix $A\in\mathbb{R}^{d\times m}$ is typically referred to as the \textit{dictionary} with $m\geq d$ elements or \textit{atoms}. It is well-known that $x$ can be uniquely identified if $A$ satisfies the so called \textit{spark condition}\footnote{That is every $2k\leq d$ columns of $A$ are linearly independent.}. Meanwhile, there exist tractable and efficient convex relaxations of the combinatorial problem of finding the (unique) $k$-sparse solution of $y=Ax$ with provable recovery guarantees \cite{r1}. 

A related problem is \textit{dictionary learning} or \textit{sparse coding} \cite{r1.9} which can be expressed as a sparse factorization \cite{r2} of the data matrix $Y=AX$ (where both $A$ and $X\in\mathbb{R}^{m\times n}$ are assumed unknown) given that each column of $X$ is $k$-sparse and $A$ satisfies the spark condition as before. A crucial question is how many data samples ($n$) are needed to \textit{uniquely} identify $A$ and $X$ from $Y$? Unfortunately, the existing lower bound is (at best) exponential $n\geq (k+1){m\choose k}$ assuming an equal number of data samples over each $k$-sparse support pattern in $X$ \cite{r3,r4}. 

In this paper, we address a more challenging problem. In particular, we are interested in the above sparse matrix factorization problem $Y=AX$ (with both sparsity and spark conditions) when only $p<d$ random linear measurements from each column of $Y$ is available. We would like to find lower bounds for $n$ for the (partially observed) matrix factorization to be unique. This problem can also be seen as recovering both the dictionary $A$ and the sparse coefficients $X$ from compressive measurements of data. For this reason, this problem has been termed \textit{Blind Compressed Sensing} (BCS) before \cite{r5}, although the end-goal of BCS is the recovery of $Y$. 

\begin{definition}[Summary of Contributions] 
We start by establishing that the uniqueness of the learned dictionary over random data measurements is a sufficient condition for the success of BCS. Perfect recovery conditions for BCS are derived under two different scenarios. In the first scenario, fewer random linear measurements are available from each data sample. It is stated that having access to a large number of data samples compensates for the inadequacy of sample-wise measurements. Meanwhile, in the second scenario, it is assumed that slightly more random linear measurements are available over each data sample and the measurements are partly fixed and partly varying over the data. This measurement scheme results in a significant reduction in the required number of data samples for perfect recovery. Finally, we address the computational aspects of BCS based on the recent non-iterative dictionary learning algorithms with provable convergence guarantees to the generating dictionary.
\end{definition}

\subsection{Prior Art on BCS} 

BCS was initially proposed in \cite{r5} where it was assumed that, for a given random Gaussian sampling matrix $\Phi\in\mathbb{R}^{p\times d}$ ($p< d$), $Z=\Phi Y$ is observed. The conclusion was that, assuming the factorization $Y=AX$ is unique, $Z=B X$ factorization would also be unique with a high probability when $A$ is an orthonormal basis. However, it would be impossible to recover $A$ from $B=\Phi A$ when $p<d$. It was suggested that structural constraints be imposed over the space of admissible dictionaries to make the inverse problem well-posed. Some of these structures were sparse bases under known dictionaries, finite set of bases and orthogonal block-diagonal bases \cite{r5}. While these results can be useful in many applications, some of which are mentioned in \cite{r5}, they do not generalize to unconstrained overcomplete dictionaries. 

Subsequently, there has been a line of empirical work on showing that dictionary learning from compressive data---a sufficient step for BCS---can be successful given that a different sampling matrix is employed for each data sample\footnote{Note that the linear form $Z=BX$ is no longer valid which is possibly a reason for the lack of a theoretical extension of BCS to this case.} (i.e. each column of $Y$). For example, \cite{r6} uses a modified K-SVD to train both the dictionary and the sparse coefficients from the incomplete data. Meanwhile, \cite{r7,r8,r9} use generic gradient descent optimization approaches for dictionary learning when only random projections of data are available. The empirical success of dictionary learning with partial as well as compressive or projected data triggers more theoretical interest in finding the uniqueness bounds of the unconstrained BCS problem. 

Finally, we must mention the theoretical results presented in the pre-print \cite{r10} on BCS with overcomplete dictionaries while $X$ is assumed to lie in a structured union of disjoint subspaces \cite{r10.1}. It is also proposed that the results of this work extend to the generic sparse coding model if the `one-block sparsity' assumption is relaxed. We argue that the main theoretical result in this pre-print is incomplete and technically flawed as briefly explained here. In the proof of Theorem 1 of \cite{r10}, it is proposed that (with adjustment of notation) \textit{``assignment [of $Y$'s columns to rank-$k_\ell$ disjoint subsets] can be done by the (admittedly impractical) procedure of testing the rank of all possible ${n\choose k_\ell}$ matrices constructed by concatenating subsets of $k_\ell+1$ column vectors, as assumed in \cite{r3}''}. However, it is ignored that the entries of $Y$ are missing at random and the rank of an incomplete matrix cannot be measured. As it becomes more clear later, the main challenge in the uniqueness analysis of unconstrained BCS is in addressing this particular issue. Two strategies to tackle this issue that are presented in this paper are: 1) increasing the number of data samples and 2) designing and employing measurement schemes that preserve the low-rank structure of $Y$'s sub-matrices. 

This paper is organized as follows. In Section \ref{sec:problem}, we provide the formal problem definition for BCS. Our main results are presented in Section \ref{sec:main}. We present the proofs in Section \ref{sec:proofs}. Practical aspects of BCS are treated in Section \ref{sec:algorithm} where we explain how provable dictionary learning algorithms, such as \cite{r12}, can be utilized for BCS. Finally, we conclude the paper and present future directions in Section \ref{sec:conclusion}.

\subsection{Notation}
Our general convention throughout this paper is to use capital letters for matrices and small letters for vectors and scalars. For a matrix $X\in\mathbb{R}^{m\times n}$, $x_{ij}\in\mathbb{R}$ denotes its entry on row $i$ and column $j$, $x_i\in\mathbb{R}^m$ denotes its $i$'th column and $vec(X)\in\mathbb{R}^{mn}$ denotes its column-major vectorized format. The inner product between two matrices $A$ and $B$ (of the same sizes) is defined as $\langle A,B\rangle =trace\left( A^T B\right)$. Let $Spark(A)$ denote the smallest number of $A$'s columns that are linearly dependent. $A$ is $\mu$-coherent if $\forall i\neq j$ we have $\frac{|\langle a_i,a_j\rangle|}{\|a_i\|_2\|a_j\|_2}\leq \mu$. 
Finally, let $[m]\coloneqq\{1,2,\dots,m\}$ and let ${[m]\choose k}$ denote the set of all subsets of $[m]$ of size $k$.

\section{BCS Problem Definition}
\label{sec:problem}

Construct the data matrix $Y\in \mathbb{R}^{d\times n}$ by concatenating $n$ signal vectors $y_j\in\mathbb{R}^d$ (for $j$ from 1 to $n$). Throughout this paper, we make the following assumptions about the sampling operator and the data sparsity. It must be noted that the following assumptions over the sparse coding of $Y$ are minimal among existing sparse coding assumptions for provable uniqueness guarantees; see e.g. \cite{r3,r4}.

\begin{definition}[Linear measurement]
Suppose $p\leq d$ linear measurements are taken from each signal $y_j\in\mathbb{R}^d$ as in 
$z_j=\Phi_j y_j\in\mathbb{R}^p$ where $\Phi_j\in\mathbb{R}^{p\times d}$ is referred to as the sampling matrix. We could also represent the measurements as a linear projection of the signal onto the row-space of the sampling matrix\footnote{Note that $z_j$ can be computed from $\hat{y}_j$ using the relationship $z_j=\Phi_j\hat{y}_j$. Therefore, $\hat{y}_j$ and $z_j$ carry the same amount of information about $y_j$ given the sampling matrix $\Phi_j$.}:
$$\hat{y}_j=\Phi_j^T \left(\Phi_j\Phi_j^T\right) ^{-1} z_j$$

We will use $\mathcal{M}^p(Y)=[z_1^T,\dots,z_n^T]^T\in\mathbb{R}^{pn}$ to denote the observations and $\mathcal{P}^p(Y)\in\mathbb{R}^{d\times n}$ to denote the projected matrix that is a concatenation of all $\hat{y}_j$. Specifically, when entries of each $\Phi_j$ are drawn independently from a random Gaussian distribution with mean zero and variance $1/d$, we use the notations $\mathcal{M}^p_G(Y)$ and $\mathcal{P}^p_G(Y)$.  
\end{definition}

\begin{definition}[Sparse coding model]
Assume $Y=AX$ where $A\in\mathbb{R}^{d\times m}$ denotes the dictionary ($m>d$ in the overcomplete setting) and $X\in\mathbb{R}^{m\times n}$ is a sparse matrix with exactly $k$ non-zero entries per column and $Spark(A)>2k$. Additionally, assume that each column of $X$ is randomly drawn by first selecting its support 
$S\in{[m]\choose k}$ uniformly at random and then filling the support entries with random i.i.d. values uniformly drawn from a bounded interval, e.g. $(0,1]\subset\mathbb{R}$. We denote by $\mathcal{Y}^m_k$ the set of feasible $Y$ under the described sparse coding model. Note that the assumption $Spark(A)>2k$ is necessary to ensure a unique $X$ even when $A$ is known and fixed. 
\end{definition}

\begin{remark}
As noted and proved in \cite{r4}, when $Y\in\mathcal{Y}^m_k$, with probability one, no subset of $k$ (or less) columns of $Y$ is linearly dependent. Also with probability one, if a subset of $k+1$ columns of $Y$ are linearly dependent, then all of the $k+1$ columns must have the same support. 
\end{remark}

Given the above definitions, we can now formally express the problem definition for BCS:

\begin{definition}[BCS problem definition]
Recover $Y\in\mathcal{Y}^m_k$ from $\mathcal{M}^p(Y)$ given $\mathcal{M}^p$, $m$ and $k$.
\end{definition}

Our results throughout this paper are mainly developed for the class of Gaussian measurements 
$\mathcal{M}^p=\mathcal{M}^p_G$. 
However, it is not difficult to extend these results to the larger class of continuous sub-Gaussian distributions for $\mathcal{M}^p$.

\section{Main Results}
\label{sec:main}


To start with, assume that there are exactly $\ell$ columns in $X$ for each support pattern $S\in\mathcal{S}$ where $\mathcal{S}={[m]\choose k}$. For better understanding and without loss of generality, one can assume that the data samples are ordered according to the following sketch for $X$:

\vspace{10pt}
\[
X_{m\times n} =
\begin{pmatrix}[c c c|c|c c c]
\noalign{\vspace{-2\normalbaselineskip}}
\multicolumn{7}{c}{n=\ell|\mathcal{S}|\mbox{ samples}}\\
\multicolumn{7}{c}{$\downbracefill$}\\
\multicolumn{3}{c}{\mbox{Group }\#1\;\;\;}&
\multicolumn{1}{c}{}&
\multicolumn{3}{c}{\mbox{Group }\#|\mathcal{S}|\;\;\;}\\
\multicolumn{3}{c}{\mbox{Support: }S^{1}}&
\multicolumn{1}{c}{}&
\multicolumn{3}{c}{\mbox{Support: }S^{|\mathcal{S}|}} \\
x_1 & \dots & x_\ell & \dots & x_{n-\ell +1} & \dots & x_n \\
\multicolumn{3}{c}{$\upbracefill$}&
\multicolumn{1}{c}{}&\multicolumn{3}{c}{$\upbracefill$} \\
\multicolumn{3}{c}{\ell\mbox{ samples}}&
\multicolumn{1}{c}{}&
\multicolumn{3}{c}{\ell\mbox{ samples}} \\
\noalign{\vspace{-2\normalbaselineskip}}
\end{pmatrix} 
\vspace{2\normalbaselineskip}
\]

The best known bound for $\ell$, for the factorization $Y=AX$ to be unique (with probability one) under the specified random sparse coding model, is $\ell\geq k+1$. This results in an exponential sample complexity $n\geq (k+1){m\choose k}$. Specifically, it is said that `$Y=AX$ factorization is unique' if there exist a diagonal matrix $D\in\mathbb{R}^{m\times m}$ and a permutation matrix $P$ such that for any other feasible factorization $Y=A'X'\in\mathcal{Y}^m_k$, we have $A'=APD$. Clearly, this ambiguity makes it more challenging to prove the uniqueness of the dictionary learning problem. Meanwhile, authors in \cite{r4} propose a strategy for handling the permutation and scaling ambiguity which is reviewed in Lemma \ref{lem:APD}.

Through the following lemma, we can establish that the uniqueness of the learned dictionary is a sufficient condition for the success of BCS (proof is provided in Appendix).

\begin{lemma}
Suppose for every pair $AX,A'X'\in\mathcal{Y}^m_k$ that satisfy 
$\mathcal{M}^p_G(A'X')=\mathcal{M}^p_G(AX)$ with $p>2k$, $A'=APD$ for some diagonal matrix $D$ and permutation matrix $P$. Then $A'X'=AX$ with probability one.
\label{lem:forgot}
\end{lemma}

Briefly speaking, existing uniqueness results exploit the fact that the rank of each group of columns in the above sketch is bounded above by $k$. This makes it possible to uniquely identify groups of samples that share the same support pattern. Meanwhile, when only $\mathcal{M}^p(Y)$ is available, it might not be possible to uniquely identify these groups. Nevertheless, it is noted in \cite{r4} that $\ell\geq k|\mathcal{S}|=k{m\choose k}$ ensures uniqueness without the need for grouping, at the cost of significantly increasing the required number of data samples (compared to $\ell\geq k+1$). 

In our initial BCS uniqueness result, we use the pigeon-hole strategy of \cite{r4} which results in a less practical bound $n\geq k|\mathcal{S}|^2$ even when $Y$ is completely observed\footnote{Authors in \cite{r4} propose a deterministic approach using the pigeon-hole principle as well as a probabilistic approach with smaller bounds for $n$.}. Yet, it is interesting to explore the implications of a finite $n$ that ensures a successful BCS for the general sparse coding model. The CS theory requires the complete knowledge of $A$ to uniquely recover $X$ and $Y$ from $\mathcal{M}^p(Y)$. Meanwhile, our results assert that $A$, $X$ and $Y$ can be uniquely identified from $\mathcal{M}^p(Y)$ given a large but finite number of samples $n$. Necessary proofs for the results of this section are presented in the following section.

\begin{theorem}
Assume $p>2k$ and there are exactly $\ell$ columns in $X$ for each $S\in\mathcal{S}$. Then $Y\in\mathcal{Y}^m_k$ can be perfectly recovered from $\mathcal{M}^p_G(Y)$ with probability one given that 
$\ell\geq \frac{2k(d-2k)+1}{p-2k} {m\choose k}$.
\label{lem:Gell-2k}
\end{theorem}

\begin{corollary}
With probability at least $1-\beta$, $Y\in\mathcal{Y}^m_k$ can be perfectly recovered from $\mathcal{M}^p_G(Y)$ given that $p>2k$ and
$n\geq \frac{2k(d-2k)+1}{\beta(p-2k)}{m\choose k}^2$.
\label{cor:Gn-2k}
\end{corollary}



Aside from the intellectual implications of Theorem \ref{lem:Gell-2k} and Corollary \ref{cor:Gn-2k} discussed above, the stated bounds for $\ell$ and $n$ are clearly not very practical. To reduce these bounds while guaranteeing the success of BCS, we introduce a \textit{hybrid measurement} scheme that we explain below.

\subsection{BCS with hybrid measurements}

\begin{definition}
{\bf (Hybrid Gaussian Measurement)} 
In a hybrid measurement scheme, $\Phi_j^T=\left[ F^T, V_j^T \right]$ 
where $F\in\mathbb{R}^{p_f\times d}$ stands for the fixed part of sampling matrix and 
$V_j\in\mathbb{R}^{p_v\times d}$ stands for the varying part of the sampling matrix. The total number of measurements per column is $p=p_f+p_v\leq d$. In a hybrid Gaussian measurement scheme, $F$ and $V_1$ through $V_n$ are assumed to be drawn independently from an i.i.d. zero-mean Gaussian distribution with variance $1/d$. The observations corresponding to $F$ and $V_j$'s are denoted by $FY\in\mathbb{R}^{p_f\times n}$ and $\mathcal{M}^{p_v}_G(Y)\in\mathbb{R}^{p_vn}$ respectively.
\end{definition}

As mentioned earlier, the hybrid measurement scheme was designed to reduce the required number of data samples for perfect BCS recovery. In particular, as formalized in Lemma \ref{lem:rank-check}, the fixed part of the measurements is designed to retain the low-rank structure of each $k$-dimensional subspace associated with a particular $S\in\mathcal{S}$. Meanwhile, the varying part of the measurements is essential for the uniqueness of the learned dictionary. 

\begin{theorem}
Assume $p>3k+1$ and there are exactly $\ell$ columns in $X$ for each $S\in\mathcal{S}$. Then $Y\in\mathcal{Y}^m_k$ can be perfectly recovered from hybrid Gaussian measurements with probability one given that 
$\ell\geq \frac{2k(d-2k)+1}{p-3k-1}$.
\label{cor:Gell-4k}
\end{theorem}


\begin{remark}
Similar to the statement of Corollary \ref{cor:Gn-2k}, it can be stated that BCS with hybrid Gaussian measurement succeeds with probability at least $1-\beta$ given that 
$n\geq \frac{2k(d-2k)+1}{\beta(p-3k-1)}{m\choose k}$. The proof follows the proof of Corollary \ref{cor:Gn-2k}.
\end{remark}

\begin{remark}
Although we mainly follow the stochastic approach of \cite{r4} in this paper, we could also employ the deterministic approach of \cite{r3} to arrive at the uniqueness bound in Theorem \ref{cor:Gell-4k}. In \cite{r3}, an algorithm (which is not necessarily practical) is proposed to uniquely recover $A$ and $X$ from $Y$. This algorithm starts by finding subsets of size $\ell$ of $Y$'s columns that are linearly dependent by testing the rank of every subset. Dismissing the degenerate possibilities\footnote{Degenerate instances of $X$ are dismissed by adding extra assumptions in the deterministic sparse coding model. Meanwhile, as pointed out in \cite{r4}, such degenerate instances of $X$ would have a probability measure of zero in a random sparse coding model}, these detected subsets would correspond to samples with the same support pattern in $X$. Under the assumptions in Theorem \ref{cor:Gell-4k}, it is possible to test whether $\ell$ columns in $Y$ are linearly dependent (with probability one), as a consequence of Lemma \ref{lem:rank-check} in the following section. 
\end{remark}

Until now, our goal was to show that $A$ (and subsequently $X$) is unique given only CS measurements. As we mentioned before, uniqueness of $A$ is a \textit{sufficient} condition for the success of BCS. Consider the scenario where not all support patterns $S\in\mathcal{S}$ are realized in $X$ or for some there is not enough samples to guarantee recovery. For such scenarios, we present the following theorem.

\begin{theorem}
Assume $p>3k+1$ and let 
$$\hat{\mathcal{S}}=\{S|S\in\mathcal{S},|J(S)|\geq\gamma\}\subseteq \mathcal{S}$$
where $\gamma=\frac{2k(d-2k)+1}{p-3k-1}$ and $J(S)$ denotes the set of indices of columns of $X$ with support $S$.
Then, under hybrid Gaussian measurement, $Y_{J(S)}$ for all $S\in\hat{\mathcal{S}}$ can be perfectly recovered with probability one.
\label{the:latest}
\end{theorem}

\section{Proofs}
\label{sec:proofs}

The following crucial lemma from \cite{r4} handles the permutation ambiguity of sparse coding.

\begin{lemma}[\cite{r4}, Lemma 1]
Assume $Spark(A)>2k$ for $A\in\mathbb{R}^{d\times m}$ and let $\mathcal{S}={[m]\choose k}$. If there exists a mapping $\pi:\mathcal{S} \rightarrow \mathcal{S}$ such that
$$span\left\{ A_{S} \right\} =span\left\{ A'_{\pi(S)} \right\}
\mbox{ for every } S\in \mathcal{S}$$  
then there exist a permutation matrix $P$ and a diagonal matrix $D$ such that $A'=APD$.
\label{lem:APD}
\end{lemma}


The following lemma from random matrix theory, along with Lemma \ref{lem:APD}, are the main ingredients of our first main result (proof is provided in the Appendix).

\begin{lemma} 
Assume $A,B\in\mathbb{R}^{d\times \ell}$ are rank-$k$ matrices and 
$\mathcal{M}_G^p$  is a Gaussian measurement operator with $p\geq (2k(d+\ell-2k)+1)/\ell$. 
If $\mathcal{M}^p_G(A)=\mathcal{M}^p_G(B)$, then $A=B$ with probability one. 
\label{lem:manifold}
\end{lemma}

\begin{proof}[Proof of Theorem \ref{lem:Gell-2k}]

Assume $A'X'$ is an alternate factorization that satisfies 
$A'X'\in\mathcal{Y}^m_k$ and $\mathcal{M}^p_G(A'X')=\mathcal{M}^p_G(AX)$. We will prove $A'=APD$ for some diagonal $D$ and some permutation matrix $P$ using Lemma \ref{lem:APD}.
Consider a particular support pattern $S\in\mathcal{S}$ and let $J(S)\subset [n]$ denote the set of indices of $X$'s columns that have the sparsity pattern $S$. 
By definition,  $|J(S)|=\ell\geq k'{m\choose k}$ where $k'=(2k(d-2k)+1)/(p-2k)$. 
Due to the pigeon-hole principle, there must be at least $k'$ columns within $X'_{J(S)}$ that share some particular support pattern $S'\in\mathcal{S}$. In other words, if $J'(S')$ denotes the set of indices of $X'$'s columns that have the support pattern $S'$, then
$|J(S)\cap J'(S')|\geq k'$. For simplicity, denote $I=J(S)\cap J'(S')$. Clearly, 
$rank(AX_I)=rank(A'X'_I)=k$ (because $|S|=|S'|=k$), and we have
$$\mathcal{M}^p_G(A'X'_I)=\mathcal{M}^p_G(AX_I)$$

According to Lemma \ref{lem:manifold}, if $p\geq (2k(d+k'-2k)+1)/k'$ or 
equivalently $k'\geq (2k(d-2k)+1)/(p-2k)$, then $A'X'_I=AX_I$ with probability one. Meanwhile, since $|I|\geq k'\geq k+1$, $A'X'_I=AX_I$ necessitates that 
\begin{equation}
span\left\{ A_S\right\} =span\left\{ A'_{S'}\right\}
\label{eq:map2}
\end{equation}
Finally, since $A$ satisfies the spark condition, it is not difficult to see that $\pi(S)=S'$ is a bijective map. To explain more, assume there exists some $S''\neq S$ such that
$$span\left\{ A_{S''}\right\} =span\left\{ A'_{S'}\right\}$$
Combining with (\ref{eq:map2}) we arrive at
$$span\left\{ A_S\right\} =span\left\{ A_{S''}\right\}\mbox{,}$$
which contradicts the spark condition for $A$ for $S''\neq S$. Therefore, $\pi$ must be injective. Now, since $\mathcal{S}$ is a finite set and $\pi$ is an injective mapping from $\mathcal{S}$ to itself, it must also be surjective and, thus, bijective. 
\end{proof}

In order to have at least $\ell$ columns in $X$ for each support $S\in\mathcal{S}$ in the random sparse coding model 
$\mathcal{Y}^m_k$, we must have more than just $n=\ell|\mathcal{S}|$ data samples. The following result from \cite{r4} quantifies the number of required data samples to ensure at least $\ell$ columns per each $S\in\mathcal{S}$ with a tunable probability of success.

\begin{lemma}[\cite{r4}, \S IV]
For a randomly generated $X$ with $n=\ell{m\choose k}$ and $\beta\in [0,1]$, with probability at least $1-\beta$, there are at least $\beta \ell$ columns for each support pattern $S\in\mathcal{S}$.
\label{lem:beta}
\end{lemma}

\begin{proof}[Proof of Corollary \ref{cor:Gn-2k}]
Proof is fairly trivial. According to Lemma \ref{lem:beta}, we need 
$n\geq \frac{\ell}{\beta}{m\choose k}$ samples to guarantee that with probability at least $1-\beta$ there are at least $\ell$ samples in $X$ for each support pattern $S\in\mathcal{S}$. In Theorem \ref{lem:Gell-2k} we established that 
$\ell\geq \frac{2k(d-2k)+1}{p-2k} {m\choose k}$ guarantees the success of BCS under Gaussian sampling. 
Therefore, 
$n\geq \frac{2k(d-2k)+1}{\beta(p-2k)} {m\choose k}^2$
guarantees the desired uniqueness.
\end{proof}

In order to prove the results for the hybrid measurement scheme, we present the following lemma which is proved in the Appendix.

\begin{lemma}
Assume $F\in\mathbb{R}^{p_f\times d}$ is drawn from an i.i.d. zero-mean Gaussian distribution (with $p_f\leq d$). Let $Y_J\in\mathbb{R}^{d\times |J|}$ denote the columns of $Y$ indexed by the set $J$. If $rank(F Y_J)=k<p_f$, then $rank(Y_J)=k$ with probability one. 
\label{lem:rank-check}
\end{lemma}

\begin{proof}[Proof of Theorem \ref{cor:Gell-4k}]
Assume $A'X'$ is an alternate factorization that satisfies 
$A'X'\in\mathcal{Y}^m_k$, $\mathcal{M}^{p_v}_G(A'X')=\mathcal{M}^{p_v}_G(AX)$ and $FA'X'=FAX$. Also assume $p_f=k+1$ and $p_v=p-k-1$.
Consider a particular support pattern $S'\in\mathcal{S}$ and let $J'(S')\subset [n]$ denote the set of indices of $X'$'s columns that have the same sparsity pattern $S'$. 

Clearly, 
$$rank\left(FA'X'_{J'(S')}\right)\leq rank\left(A'X'_{J'(S')}\right)=k$$
Therefore, if $p_f \geq k+1$, then 
$p_f > rank\left(FA'X'_{J'(S')}\right)$ and according to Lemma \ref{lem:rank-check}:
$$rank\left(FA'X'_{J'(S')}\right)=rank\left(A'X'_{J'(S')}\right)=k$$ 
with probability one. Hence, $rank\left(FAX_{J'(S')}\right)=k$.
Again using Lemma \ref{lem:rank-check} with $p_f \geq k+1$,  
$$rank\left(AX_{J'(S')}\right)=rank\left(FAX_{J'(S')}\right)=k$$
with probability one. Therefore, all the columns in $X_{J'(S')}$ must have the same support, namely $S$. Note that since $J'(S')\subseteq J(S)$, $|J'(S')|\leq |J(S)|=\ell$. Meanwhile, 
$$\sum_{S'\in\mathcal{S}}|J'(S')|=\ell{m\choose k}$$
necessitates that $|J'(S')|= \ell$ for every $S'\in\mathcal{S}$. Therefore, 
$|J(S)\cap J'(S')|=|I|=\ell$. Now, given
$$\mathcal{M}^{p_v}_G(A'X'_I)=\mathcal{M}^{p_v}_G(AX_I)$$
according to Lemma \ref{lem:manifold}, if $\ell\geq (2k(d-2k)+1)/(p_v-2k)$, then $A'X'_I=AX_I$ with probability one. Meanwhile, since $|I|=\ell \geq k+1$, $A'X'_I=AX_I$ necessitates that 
\begin{equation}
span\left\{ A_S\right\} =span\left\{ A'_{S'}\right\}
\label{eq:map}
\end{equation}
Finally, since $A$ satisfies the spark condition, $\pi(S)=S'$ is a bijective map and $A'=APD$ for some diagonal $D$ and permutation matrix $P$ according to Lemma \ref{lem:APD}. 
\end{proof}

\begin{proof}[Proof of Theorem \ref{the:latest}]
Recall that for every $S\in\hat{\mathcal{S}}$ we have
$|J(S)|\geq \gamma \geq k+1$. Assume $p_f=k+1$ and $p_v=p-k-1$ as before. Having $p_f\geq k+1$ allows testing whether a subset of $k+1$ columns of $Y$ are linearly dependent (have a rank of $k$) with probability one. Therefore, by doing an exhaustive search among every 
sub-matrix $Y_J$ with $J\in{[n]\choose k+1}$, we are able to find subsets of $J(S)$ (of size $k+1$) if $|J(S)|\geq k+1$. Moreover, we can combine and complete these subsets to uniquely identify every rank-$k$ sub-matrix $Y_{J(S)}$ with $|J(S)|\geq k+1$. 

Now, among these sub-matrices, those with $|J(S)|\geq \gamma$ can be recovered perfectly (with probability one) since, for any rank-$k$ matrices $Y_{J(S)}$ and $\hat{Y}_{J(S)}$,
$$\mathcal{M}^{p_v}_G(Y_{J(S)})=\mathcal{M}^{p_v}_G(\hat{Y}_{J(S)})$$
with 
$$p_v\geq (2k(d+|J(S)|-2k)+1)/|J(S)|$$ 
or 
$|J(S)|\geq \frac{2k(d-2k)+1}{p_v-2k}$
implies 
$Y_{J(S)}=\hat{Y}_{J(S)}$ according to Lemma \ref{lem:manifold}.
\end{proof}

\section{Algorithmic Performance of BCS under Hybrid Measurements}
\label{sec:algorithm}

Recall that in the dictionary learning (DL) problem, the data matrix $Y\in\mathbb{R}^{d\times n}$ is given where $Y=A^* X^*\in\mathcal{Y}^m_k$ and the task is to factorize $Y=AX\in\mathcal{Y}^m_k$ such that $A=A^* PD$ for some permutation matrix $P$ and diagonal matrix $D$. Unfortunately, the corresponding optimization problem is non-convex (even with $\ell_1$ relaxation). The majority of existing DL algorithms are based on the iterative scheme of starting from an initial state $Y=A^{(0)}X^{(0)}$ and alternating between updating $X^{(t+1)}$ while keeping $A^{(t)}$ fixed and updating $A^{(t+1)}$ while keeping $X^{(t+1)}$ fixed, each corresponding to a convex problem. It has been recently shown that if the initial dictionary $A^{(0)}$ is sufficiently close\footnote{The \textit{basin of attraction} has a swath of $O(k^{-2})$ \cite{r15}.} to $A^* PD$ for some $P$ and $D$, then the iterative algorithm converges to $A^* PD$ under certain incoherency assumptions about $A^*$ \cite{r15}. Similar guarantees have been derived for the well-known K-SVD algorithm 
\cite{r18}.  

Furthermore, DL from incomplete or corrupt data has also been tackled in several studies. In particular, DL from compressive measurements has been addressed in \cite{r6,r7,r8,r9} where different iterative DL algorithms are modified to accommodate the compressive measurements. In some cases, these modifications have been justified by showing that the output of each iteration does not significantly deviate from the reference output based on the complete data. However, to best of our knowledge, there are no convergence guarantees to $A^* PD$ for these iterative algorithms. As we mentioned before, a successful DL from compressive measurements is a sufficient condition for a successful BCS. In this section, we plan to investigate the utility of a recently proposed (non-iterative) DL algorithm \cite{r12} with guarantees for the approximate recovery of $A^* PD$ for an incoherent $A^*$. One would hope that $A^* PD$ can be approximated from $Y$ with fewer data samples than is required for the exact identification of $A^* PD$ which was the topic of previous sections.

Below, we review the main result of \cite{r12} and analyze the performance of their DL algorithm if only hybrid Gaussian measurements were available. Recall that in our BCS measurement scheme, $p_f$ fixed and $p_v$ varying linear measurements are taken from each sample for a total of $p=p_f+p_v$ linear measurements (per sample). Before presenting their result, we need to introduce some new notation as well as modifications to the sparse coding model to reflect the model used in \cite{r12}. In particular, let $\mathcal{X}\in\mathbb{R}^m$ denote the random vector of sparse coefficients where its distribution class $\Gamma$ is defined below. Hence, each $x_j$ denotes an outcome of $\mathcal{X}$. Also, let $\mathcal{X}_i$ denote the random variable associated with the $i$'th entry of $\mathcal{X}$. 

\begin{definition}
(Distribution class $\Gamma$)
The distribution is in class $\Gamma$ if $i$) 
$\forall \mathcal{X}_i\neq 0\colon \mathcal{X}_i\in[-C,-1]\cup [1,C]$ and 
$\mathbb{E}[\mathcal{X}_i]=0$. $ii$) Conditioned on any subset of coordinates in $\mathcal{X}$ being non-zero, the values of $\mathcal{X}_i$ are independent of each other. Distribution has \textit{bounded $\ell$-wise moments} if the probability that $\mathcal{X}$ is non-zero in any subset $S$ of $\ell$ coordinates is at most $c^\ell$ times 
$\prod_{i\in S} \mathbb{P}[X_i\neq 0]$ where $c=O(1)$.
\end{definition}

\begin{remark}
Similar to \cite{r12}, in the rest of paper we will assume $C=1$. Derived results generalize to the case $C>1$ by loosing constant factors in guarantees.
\end{remark}

\begin{definition}
Two dictionaries $A,B\in\mathbb{R}^{d\times m}$ are column-wise $\epsilon$-close, if there exists a permutation $\pi$ and $\theta\in\{\pm 1\}^m$ such that $\forall i\in[m]\colon \|a_i-\theta_i b_{\pi(i)}\|_2\leq \epsilon$.
\end{definition}

\begin{remark}
When talking about two dictionaries $A$ and $B$ that are $\epsilon$-close, we always assume the columns are ordered and scaled correctly so that $\|a_i-b_i\|_2\leq \epsilon$.
\end{remark}

\begin{theorem}[\cite{r12}, Theorem 1.4]
There is a polynomial time algorithm to learn a $\mu$-coherent dictionary $A$ from random samples. With high probability, the algorithm returns a dictionary $\hat{A}$ that is column-wise $\epsilon$-close to $A$ given random samples of the form $\mathcal{Y}=A\mathcal{X}$, where $\mathcal{X}$ is drawn from a distribution in class $\Gamma$. Specifically, if $k\leq c\min(m^{(\ell-1)/(2\ell-1)},1/(\mu\log d))$ and the distribution has bounded $\ell$-wise moments, $c>0$ is a constant only depending on $\ell$, then the algorithm requires 
$n=\Omega ((m/k)^{\ell-1}\log m+m k^2\log m \log 1/\epsilon)$ samples and runs in time $\tilde{O}(n^2d)$.
\label{the:arora}
\end{theorem}

\begin{definition}[Summary of the algorithm of \cite{r12}] 
This algorithm, which has fundamental similarities with a concurrent work \cite{r13}, consists of two main stages:
$i$) \textit{Data Clustering}: the connection graph is built where each node corresponds to a column of $Y$ and an edge between $y_i$ and $y_j$ implies their supports $S_i$ and $S_j$ have a non-empty intersection. Then, an overlapping clustering procedure is performed over the connection graph to find overlapping maximal cliques (with missing edges). $ii$) \textit{Dictionary Recovery}: every cluster in the connection graph represents the set of samples associated with a single dictionary atom. After finding these clusters in the connection graph, each atom is approximated by the principal eigenvector of the covariance matrix for the data samples in its corresponding cluster.
\end{definition}

There are two challenges in extending the above result to the BCS framework: $i$) during generation of the connection graph from data and $ii$) during computation of the principal eigenvector of the data covariance matrix. We address these challenges separately in the following subsections.

\subsection{Building the data connection graph for BCS}

For building the connection graph, we use the fixed part of the hybrid measurements, i.e. $FY$ with $F\in\mathbb{R}^{p_f\times d}$ drawn from a Gaussian distribution. Computation of the connection graph in \cite{r12} relies on the following lemma.

\begin{lemma}[\cite{r12}, Lemma 2.2]
Suppose $k<1/(C'\mu \log d)$ for large enough $C'$ (depending on $C$ in the definition of $\Gamma$). Then, if $S_i$ and $S_j$ are disjoint, with high probability $|\langle y_i, y_j \rangle|<1/2$.
\end{lemma} 

Without going into the details of the clustering algorithm of \cite{r12}, we study the conditions under which the connection graph does not change when only $p_f$ linear measurements from each data sample is given. Let
$FA\in\mathbb{R}^{p_f\times m}$ be $\mu_f$-coherent. It is not hard to see from the above lemma that if $k<1/(C'\mu_f \log d)$, then with high probability for disjoint $S_i$ and $S_j$, $|\langle F y_i,F y_j \rangle|<1/2$. To establish a relationship between $\mu_f$, $\mu$ and $p_f$, we use the following result from \cite{r22}.

\begin{lemma}[\cite{r22}, Lemma 3.1]
Let $x,y\in\mathbb{R}^d$ with $\|x\|_2,\|y\|_2\leq 1$. Assume $\Phi\in\mathbb{R}^{n\times d}$ is a random matrix with independent $\mathcal{N}(0,1/n)$ entries. Then, for all $t>0$
$$\mathbb{P}[|\langle \Phi x,\Phi y\rangle-\langle x,y\rangle|\geq t]\leq 2\exp (-n\frac{t^2}{C_1+C_2t})$$
with $C_1=\frac{8e}{\sqrt{6\pi}}\approx 5.0088$ and $C_2=\sqrt{8}e\approx 7.6885$.
\label{lem:holger}
\end{lemma}

\begin{corollary}
Assume $F\in\mathbb{R}^{p_f\times d}$ has i.i.d entries from $\mathcal{N}(0,1/p_f)$. Let $A$ be $\mu$-coherent and $FA$ be $\mu_f$-coherent. Then,
$$\mathbb{P}[\mu_f\geq \mu+t]\leq 2\exp (-p_f\frac{t^2}{C_1+C_2t})$$
with $C_1$ and $C_2$ specified in Lemma \ref{lem:holger}.
\label{cor:mu}
\end{corollary}

\begin{proof}
Note that the variance of $F$'s entries does not have an effect on $\mu_f$ due to the normalization in the definition of the coherency and we could assume $F$'s entries have variance $1/d$ as before. We exploit Lemma \ref{lem:holger} by replacing $x=a_i$ and $y=\pm a_j$ and $\Phi=F$. Proof is complete by noticing that 
$\mathbb{P}[\mu_f\geq \mu+t]\leq\mathbb{P}[|\mu_f-\mu|\geq t]$
\end{proof}

Based on Corollary \ref{cor:mu}, it can be deduced that with high probability $\mu_f\leq \mu+\sqrt{\log(p_f)/p_f}$. Therefore, replacing $k<c/(\mu \log d)$ in the original Theorem \ref{the:arora} with $k<c/(\mu_f \log d)$ introduces slightly stronger sparsity requirement for the success of the algorithm.

\subsection{Dictionary estimation for BCS}

At this stage, we only exploit the varying part of the measurements $\mathcal{M}^{p_v}_G(Y)$ and use $p$ in place of $p_v$ for simplicity. Let $\mathcal{C}_1,\mathcal{C}_2,\dots,\mathcal{C}_m$ be the $m$ discovered overlapping clusters from the previous stage and define the empirical covariance matrix 
$\hat{\Sigma}_i=\frac{1}{|\mathcal{C}_i|}\sum_{y_j\in \mathcal{C}_i} y_j y_j^T$ for the cluster $i$. The SVD approach\footnote{In fact, \cite{r12} proposes two methods for dictionary estimation: $i$) selective averaging and $ii$) the SVD-based approach. We selected to work with the SVD approach due to its more abstract and versatile nature.} of \cite{r12} estimates $a_i$ by $\hat{a}_i$ which is the principal eigenvector\footnote{The principal eigenvector is equivalent to the first singular vector of the covariance matrix.} of $\hat{\Sigma}_i$. Let 
$$\tilde{\Sigma}_i=\frac{1}{|\mathcal{C}_i|}\sum_{\hat{y}_j\in \mathcal{C}_i} \hat{y}_j \hat{y}_j^T$$
denote the empirical covariance matrix resulting from the compressive measurements where 
$\hat{y}_j=\Phi_j^T(\Phi_j\Phi_j^T)^{-1}\Phi_j y_j$ as before. Similarly, let $\tilde{a}_i$ denote the principal eigenvector of $\tilde{\Sigma}_i$. Our goal in this section is to show that $\|\tilde{a}_i-\hat{a}_i\|_2$ is bounded by a small constant for finite $n$ and approaches zero for large $n$. For this purpose, we use the recent results from the area of \textit{subspace learning}, specifically, subspace learning from compressive measurements \cite{r14.1}. A critical factor in estimation accuracy of the principle eigenvector of a \textit{perturbed} covariance matrix is the \textit{eigengap} between the principal and the second eigenvalues of the original covariance matrix. This is a well-known result from the works of Chandler Davis and William Kahan known as the Davis-Kahan sine theorem \cite{r14.2}. 

Consider the following notation. Let $\hat{\Pi}_k$ and $\tilde{\Pi}_k$ denote projection operators onto the principal $k$-dimensional subspaces of $\hat{\Sigma}$ and $\tilde{\Sigma}$ respectively (i.e. the projection onto the top-$k$ eigenvectors). Let $\|\tilde{\Pi}_k-\hat{\Pi}_k\|_2$ denote the spectral norm of the difference between $\hat{\Pi}_k$ and $\tilde{\Pi}_k$. Define the eigengap $\hat{\gamma}_k$ as the distance between the $k$'th and $k+1$'st largest eigenvalues of $\hat{\Sigma}$. Suppose $\hat{\Sigma}$ is computed from at least $\ell$ data samples ($|\mathcal{C}_i|\geq \ell$ for all $i$). Moreover, assume the data samples have bounded $\ell_2$ norms, i.e. $\forall j\in[\ell]\colon\|y_j\|_2^2\leq \eta$ for some positive $\eta\in\mathbb{R}$. 

\begin{lemma}[\cite{r14.1}, Theorem 1]
With probability at least $1-\delta$
$$\|\hat{\Pi}_k-\tilde{\Pi}_k\|_2 \leq \frac{1}{\hat{\gamma}_k} \left(\sqrt{\frac{88\eta^2}{\ell p}\log (d/\delta)}+
\frac{8}{3}\frac{\eta d^2}{p^2 \ell}\log (d/\delta) \right) $$
so that one can achieve $\|\hat{\Pi}_k-\tilde{\Pi}_k\|_2\leq\epsilon$ provided that
$$\ell \geq \max\left\{\frac{352\eta^2 \log(d/\delta)}{p\hat{\gamma}_k^2\epsilon^2},
\frac{16}{3}\frac{\eta d^2}{\hat{\gamma}_k\epsilon p^2}\log(d/\delta) \right\} $$
\label{lem:akshay}
\end{lemma}

Below, we present a customization of Lemma \ref{lem:akshay} for the $\ell_2$ error of the principal eigenvector estimator.

\begin{corollary}
Let $\hat{a}_i$ and $\tilde{a}_i$ represent the principal eigenvectors of $\hat{\Sigma}_i$ and $\tilde{\Sigma}_i$ respectively. With probability at least $1-\delta$ for all $i\in[m]$
$$\|\hat{a}_i-\tilde{a}_i\|_2 \leq \frac{2}{\hat{\gamma}_1} 
\left(\sqrt{\frac{88\eta^2}{\ell p}\log (d/\delta)}+
\frac{8}{3}\frac{\eta d^2}{p^2 \ell}\log (d/\delta) \right) $$
\label{lem:temp}
\end{corollary}

\begin{proof}
Clearly, $\hat{\Pi}_1=\hat{a}_i \hat{a}_i^T$ and $\tilde{\Pi}_1=\tilde{a}_i \tilde{a}_i^T$. As we mentioned in the definition of $\epsilon$-closeness, $\theta_i$ is implicit in the error expression $\|\hat{a}_i-\tilde{a}_i\|_2$ requiring that $\|\hat{a}_i-\tilde{a}_i\|_2\leq \|\hat{a}_i+\tilde{a}_i\|_2$ and consequently
$\langle \hat{a}_i,\tilde{a}_i \rangle\geq 0$. 
Also note that, by definition, for any $z\in\mathbb{R}^{d}$ 
$$\frac{\|(\hat{\Pi}_1-\tilde{\Pi}_1)z\|_2}{\|z\|_2}\leq \|\hat{\Pi}_1-\tilde{\Pi}_1\|_2$$
Now let $z=\hat{a}_i+\tilde{a}_i$. Then
\begin{eqnarray*}
\frac{\|(\hat{\Pi}_1-\tilde{\Pi}_1)z\|_2}{\|z\|_2} &=&
(1+\langle \hat{a}_i,\tilde{a}_i \rangle)
\frac{\|\hat{a}_i-\tilde{a}_i\|_2}{\|\hat{a}_i+\tilde{a}_i\|_2} \\
&\geq& \frac{1}{2}\|\hat{a}_i-\tilde{a}_i\|_2
\end{eqnarray*}
Therefore
$$\|\hat{a}_i-\tilde{a}_i\|_2 \leq 2 \|\hat{\Pi}_1-\tilde{\Pi}_1\|_2$$
and the rest follows from Lemma \ref{lem:akshay}.
\end{proof}

To obtain a lower-bound for the eigengap $\hat{\gamma}_1$ we need to review some of the intermediate results in \cite{r12}. In fact, we compute a lower-bound for $\gamma_1$ of $\Sigma$ which serves as a close approximation of $\hat{\gamma}_1$ when the number of data samples $\ell$ is large. For every $i\in[m]$, let $\Gamma_i$ be the distribution conditioned on $\mathcal{X}_i\neq 0$. Let $\alpha=|\langle u,a_i \rangle|$ for any unit-norm $u$ and let 
$$R_i^2=\mathbb{E}_{\Gamma_i}[\langle a_i,\mathcal{Y} \rangle ^2] = 
1+\sum_{j\neq i} \langle a_i,a_j \rangle ^2 \mathbb{E}_{\Gamma_i}[\mathcal{X}_j^2]$$ 
denote the \textit{projected variance} of $\Gamma_i$ on the direction $u=a_i$. It is shown \cite{r12} that generally
$$\mathbb{E}_{\Gamma_i}[\langle u,\mathcal{Y} \rangle ^2] \leq 
\alpha^2 R_i^2+2\alpha\sqrt{1-\alpha^2}\zeta +(1-\alpha^2)\zeta^2$$
where $\zeta=\max\{\frac{\mu k}{\sqrt{d}},\sqrt{\frac{k}{m}} \}$. 

The principal eigenvector of $\Sigma_i$ can be computed by finding the unit-norm $u$ that maximizes $\mathbb{E}_{\Gamma_i}[\langle u,\mathcal{Y} \rangle ^2]$. Meanwhile, it has been established that for $u=a_i$, $\mathbb{E}_{\Gamma_i}[\langle u,\mathcal{Y} \rangle ^2] = R_i^2$. Therefore, the range of $\alpha$ for the principal eigenvector must satisfy the inequality (for $\alpha\leq 1$)
$$ R_i^2\leq \alpha^2 R_i^2+2\alpha\sqrt{1-\alpha^2}\zeta +(1-\alpha^2)\zeta^2$$
It is not difficult to show this range is
$$\alpha\in [\frac{R_i^2-\zeta^2}{\sqrt{4\zeta^2+(R_i^2-\zeta^2)^2}},1]$$
Now, for the second eigenvector and eigenvalue pair we must find a unit-norm $v$ that satisfies $\langle v,u\rangle=0$ and maximizes $\mathbb{E}_{\Gamma_i}[\langle v,\mathcal{Y} \rangle ^2]$. Define $\beta=|\langle v,a_i \rangle|$. It can be shown that 
$$\beta\in [-\frac{2\zeta}{\sqrt{4\zeta^2+(R_i^2-\zeta^2)^2}},
\frac{2\zeta}{\sqrt{4\zeta^2+(R_i^2-\zeta^2)^2}}]$$
Note that the first and the second largest eigenvalues correspond to projected variances of $\Gamma_i$ on the directions of $u$ and $v$, respectively. Therefore, based on the derived ranges for $\alpha$ and $\beta$, we are able to find the following lower-bound for $\gamma_1$:
$$\gamma_1\geq R_i^2-\left(\frac{3R_i^2\zeta}{R_i^2-\zeta^2}\right)^2$$
Note that $\zeta$ becomes very small as the problem size ($d$, $m$, $n$) becomes large, resulting in 
$\hat{\gamma}_1\approx\gamma_1\approx 1$. Therefore, given a sufficient number of samples, it can be guaranteed that $\tilde{a}_i$ is an accurate estimation of $\hat{a}_i$ and, in turn, an accurate estimation of $a_i$ even when only $p<d$ measurements per sample is available. Once the dictionary has been approximated to within a close distance from the optimal dictionary $A^*PD$, iterative algorithms such as \cite{r6,r7,r8,r9} can assure convergence to a local optimum and therefore perfect recovery as suggested in \cite{r15,r12,r13}. Finally, perfect recovery of the dictionary results in perfect recovery of $X$ and $Y$ given the CS bounds for the number of measurements \cite{r1} which are generally weaker than the stated bounds for the recovery of the dictionary.

\section{Conclusion}
\label{sec:conclusion}

In this work, we studied the conditions for perfect recovery of both the dictionary and the sparse coefficients from linear measurements of the data. The first part of this work brings together some of the recent theories about the uniqueness of dictionary learning and the blind compressed sensing problem. Moreover, we described a `hybrid' random measurement scheme that reduces the theoretical bounds for the minimum number of data samples to guarantee a unique dictionary and thus perfect recovery for blind compressed sensing. In the second part, we discussed the algorithmic aspects of dictionary learning under random linear measurements. It was shown that a polynomial-time algorithm can assure convergence to the generative dictionary given a sufficient number of data samples with high probability. It would be interesting to explore dictionary learning and blind compressed sensing for non-Gaussian random measurements. In particular, when the data matrix is partially observed (i.e. an incomplete matrix), data recovery becomes a matrix completion problem where the elements of the data matrix are assumed to lie in a union of interconnected rank-$k$ subspaces. This is a subject of future work.

\section*{Appendix}

\subsection{Proof of Lemma \ref{lem:forgot}}
Let $X''\coloneqq PDX'$. Note that $A'X'=APDX'=AX''$. Thus, $\mathcal{M}^p_G(AX'')=\mathcal{M}^p_G(AX)$. Our goal is to show $X''=X$ and thus $A'X'=AX''=AX$. To prove $X''=X$, we must show that for every $j\in[n]$, $\Phi_j A x''_j=\Phi_j A x_j$ results in $x''_j=x_j$ with probability one. For simplicity, we omit the sample index $j$ in the rest of the proof.

Let $S$ and $S''$ respectively denote the sets of non-zero indices of $x$ and $x''$ where $|S|,|S''|\leq k$.
Rewrite $\Phi Ax''=\Phi Ax$ as $\Phi A(x''-x)=0$. Note that $x''-x$ is supported on $T=S\cup S''$ where $|T|\leq 2k$. Therefore, we must show that, with probability one, 
\begin{equation*}
\forall T\in{[m]\choose 2k}\colon rank(\Phi A_T)=|T|
\end{equation*}
necessitating $x''-x=0$ or $x''=x$. Since $Spark(A)>2k$, every $2k$ columns of $A$ are linearly independent and we are able to perform a Gram-Schmidt orthogonalization on $A_T$ to get $A_T=UV$ where $U\in\mathbb{R}^{d\times 2k}$ is orthonormal ($d\geq 2k$) and $V$ is a full-rank square matrix. Hence, $\Phi U\in\mathbb{R}^{p\times 2k}$ is distributed according to i.i.d. Gaussian and is full-rank with probability one \cite{r27}. We conclude the proof by noticing that $rank(\Phi UV)=rank(\Phi U)=2k$ since $V$ is a full-rank square matrix. 

\subsection{Proof of Lemma \ref{lem:manifold}}

Denote a general linear matrix measurement operator 
$\mathcal{M}\colon\mathbb{R}^{d\times n} \rightarrow \mathbb{R}^\tau$ such that 
$\mathcal{M}(Y)=\zeta=\left[ \zeta_1,\zeta_2,\dots,\zeta_\tau \right]^T$, 
$\zeta_i=\langle M_i,Y\rangle$ for $i\in[\tau]$. If we denote

\begin{equation}
\Phi=\left[
\begin{array}{c}
vec(M_1)^T\\
vec(M_2)^T\\
\vdots \\
vec(M_T)^T
\end{array}
\right]\in\mathbb{R}^{\tau\times d n}
\label{eq:phi}
\end{equation}
then $\mathcal{M}(Y)=\Phi vec(Y)$. Specifically, under the Gaussian measurement scheme for BCS, we have:
\begin{equation}
\Phi^{CS}_G=\left[
\begin{array}{ccc}
\Phi_1 	& 				& \\
				& \ddots 	& \\
				& 				& \Phi_n
\end{array}
\right]\in\mathbb{R}^{p n\times d n}
\end{equation}
where non-zero entries of $\Phi^{CS}_G$ are i.i.d. Gaussian with mean zero and variance $1/d$.

The following result from \cite{r11} gives the required number of linear measurements to guarantee (with probability one) that a rank-$r$ matrix does not fall into the null-space of the measurement operator.


\begin{lemma}[\cite{r11}, Theorem 3.1]
Let $\mathcal{R}$ be a $q$-dimensional continuously differentiable manifold over the set of $d\times d$ real matrices. Suppose we take $\tau\geq q+1$ linear measurements from $Y\in\mathcal{R}$. Assume there exists a constant $C=C(d)$ such that  $\mathbb{P}(|\langle M_i,X\rangle|<\epsilon)<C\epsilon$ for every $Y$ with $\|Y\|_F=1$. Further assume that for each $Y\neq 0$ that the random variables $\{\langle M_i,Y\rangle\}$ are independent. Then with probability one,
$Null(\mathcal{M})\cap \mathcal{R}\setminus\{0\}=\emptyset$.
\label{lem:r11}
\end{lemma}

A careful inspection of the derivation of the above theorem in \cite{r11} reveals that this result can be extended to include the manifolds over the set of rectangular matrices $Y\in\mathbb{R}^{d\times n}$. Specifically, for the manifold over rank-$r$ $d\times n$ matrices we have (see \cite{r11.1} for example)
$q=\dim(\mathcal{R})=r(d+n-r)$.

The following lemma establishes a sufficient lower bound for $\tau$ to guarantee that $\mathcal{M}(A)=\mathcal{M}(B)$ results in $A=B$.

\begin{lemma}
Let $\mathcal{R}$ denote the manifold over the set of rank-$r$ $d\times n$ matrices and let $\mathcal{R}'$ denote the manifold over the set of rank-$2r$ $d\times n$ matrices. Also let $\mathcal{M}\colon\mathbb{R}^{d\times n} \rightarrow \mathbb{R}^\tau$ with $\tau\geq \dim(\mathcal{R}')+1=2r(d+n-2r)+1$. Then, for any $A,B\in\mathcal{R}$, 
$\mathcal{M}(A)=\mathcal{M}(B)$ implies $A=B$ with probability one. 
\end{lemma}

\begin{proof}
Clearly, $\tau\geq \dim(\mathcal{R}')$ implies $\tau\geq \dim(\mathcal{R}'')$ for any $\mathcal{R}''$ over the set of rank-$r''$ $d\times n$ matrices with $r''\leq 2r$. Also note that $rank(A-B)\leq 2k$, thus $A-B\in\mathcal{R}''$. Now, since $\mathcal{M}(A-B)=0$ and 
$Null(\mathcal{M})\cap \mathcal{R}''\setminus\{0\}=\emptyset$ 
(with probability one, according to Lemma \ref{lem:r11}), we must have $A-B=0$ or $A=B$ with probability one.
\end{proof}

It only remains to show that $\mathcal{M}_G^p$ satisfies the requirements of Lemma \ref{lem:r11}. As noted in \cite{r11}, the requirement $\mathbb{P}(|\langle M_i,Y\rangle|<\epsilon)<C\epsilon$ requires that the densities of 
$\langle M_i,Y\rangle$ do not spike at the origin; a sufficient condition for this to hold for every $Y$ with $\|Y\|_F=1$ is that each $M_i$ has i.i.d. entries with a continuous density. Note that non-zero entries of $\Phi_G^{CS}$ are i.i.d. Gaussian and cover every column in $Y$. Therefore, none of the entries of $\Phi_G^{CS}vec(Y)$ would spike at the origin or equivalently there exists $C=C(d,n)$ so that $\mathbb{P}(|\langle (\Phi_j^T)_i,y_j\rangle|<\epsilon)<C\epsilon$ with $\|y_j\|_2=\Omega(1/\sqrt{n})$ given that the vector $(\Phi_j^T)_i$ is drawn from a continuous distribution. 

\subsection{Proof of Lemma \ref{lem:rank-check}}

Let $r=rank(Y_J)$ and $k=rank(FY_J)$. Perform a Gram-Schmidt orthogonalization on $Y_J$ to obtain $Y_J=UV$ where $U\in\mathbb{R}^{d\times r}$ has orthogonal columns and $V\in\mathbb{R}^{r\times |J|}$ is full-rank; hence, given $r\leq |J|$, we have $k=rank(FUV)=rank(FU)$. Note that, since $U$ is orthogonal and $F$ is i.i.d. Gaussian, $FU$ is also i.i.d. Gaussian. Hence, with probability one, $FU$ is full-rank \cite{r27} and $k=\min(p_f,r)$. To conclude the proof, note that when $k<p_f$, necessarily we have $k=r$.


\end{document}